\documentclass[11pt]{amsart}

\usepackage[utf8]{inputenc}
\usepackage[english]{babel}
\usepackage{amsthm}
\usepackage{amsmath}
\usepackage{amssymb}
\usepackage{enumerate}
\usepackage{cite}
\usepackage{algorithm}
\usepackage{algpseudocode}

\usepackage{url}

\theoremstyle{plain}
\newtheorem{thm}{Theorem}
\newtheorem{lemma}{Lemma}

\theoremstyle{definition}
\newtheorem{definition}{Definition}

\theoremstyle{remark}
\newtheorem{remark}{Remark}

\newcommand{\fp}{\mathbb{F}_p}

\newcommand{\E}{\mathbf{E}}
\newcommand{\ecp}{\E(\fp)}

\newcommand{\ffp}{\fp(\E)}

\newcommand{\Z}{\mathbb{Z}}
\newcommand{\Q}{\mathbb{Q}}

\renewcommand{\c}{\mathbf{c}}
\newcommand{\e}{\mathbf{e}}

\begin{document}

\title{Predicting the elliptic curve congruential generator}

\author{L\'aszl\'o M\'erai}

\address{Johann Radon Institute for Computational and Applied Mathematics, Austrian Academy of Sciences  Austrian Academy of Sciences,
  Altenbergerstr.\ 69, 4040 Linz, Austria}

\email{merai@cs.elte.hu}

\begin{abstract}
Let $p$ be a prime and let $\E$ be an elliptic curve defined over the finite field $\fp$ of $p$ elements. For a point $G\in\ecp$ the elliptic curve congruential generator (with respect to the first coordinate) is a sequence  $(x_n)$ defined by the relation $x_n=x(W_n)=x(W_{n-1}\oplus G)=x(nG\oplus W_0)$, $n=1,2,\dots$, where $\oplus$ denotes the group operation in $\E$ and $W_0$ is an initial point. In this paper, we show that if some consecutive elements of the sequence $(x_n)$ are given as integers, 
then one can compute in polynomial time an elliptic curve congruential generator (where the curve possibly defined over the rationals or over a residue ring) such that the generated sequence is identical to $(x_n)$ in the revealed segment. It turns out that in practice,  all the secret parameters, and thus the whole sequence $(x_n)$, can be computed from eight consecutive elements, even if the prime and the elliptic curve are private.
\end{abstract}

\subjclass[2010]{Primary 11Y50, 11Y55, 11T71, 14H52, 94A60}

\keywords{elliptic curve, congruential generator, cryptography}

 \maketitle

 \let\thefootnote\relax\footnote{The final publication is available at Springer via \url{http://dx.doi.org/ 10.1007/s00200-016-0303-x}}
 
\section{Introduction}

For a prime $p$, denote by $\fp$ the field of $p$ elements and always assume that it is represented by the first $p$-many non-negative integers $\{0,1,\dots, p-1\}$.

Let $\E$ be an elliptic curve defined over $\fp$ given by an \emph{affine Weierstrass equation}, which for $\gcd(p,6)=1$ takes the form
\begin{equation}\label{eq:weierstrass}
 y^2=x^3+Ax+B,
\end{equation}
for some $A,B\in \fp$ with non-zero discriminant $4A^3+27B^2\neq 0$.

The $\fp$-rational points $\ecp$ of $\E$ form an Abelian group (with respect to the usual addition  denoted by $\oplus$) with the point at infinity $\mathcal{O}$ as the neutral element. Let $x(\cdot)$ and $y(\cdot)$ be the coordinate functions, then for a point $P\in\ecp$, $P\neq \mathcal{O}$ with affine components $P=(x_P,y_P)$ we have $x(P)=x_P$ and $y(P)=y_P$.

For a given point $G\in\ecp$ and \emph{initial point} $W_0\in\ecp$ the \emph{elliptic curve congruential generator}  is the sequence $(W_n)$ of points of $\ecp$ satisfying the recurrence relation 
\begin{equation*}
 W_n=G \oplus W_{n-1}=nG\oplus W_0, \quad n=1,2,\dots
\end{equation*}
We also define the \emph{elliptic curve congruential generator with respect to the first coordinate} as the sequence $(x_n)$ in $\fp$ as
\begin{equation}\label{eq:def-x}
 x_n=x(W_n)=x(nG\oplus W_0), \quad n=1,2,\dots
\end{equation}

The elliptic curve congruential generator has been widely studied, many positive results have been proven about its randomness, see \cite{BD,Clegendre,CGP,CLX,ElMahassniShparlinski,GBS,HessShparlinski,HHF,Liu,LWZ,merai-1,merai-2,TopWinterhof} and see also the survey paper \cite{shparlinski-EC-survey}.
In particular, El Mahassni and Shparlinski~\cite{ElMahassniShparlinski} showed that  $(W_n)$, and so the sequence $(x_n)$,  is well-distributed.  Hess and Shparlinski~\cite{HessShparlinski}, and Topuzo\u glu and Winterhof~\cite{TopWinterhof} provided lower bounds to the linear complexity profile of the sequence $(x_n)$.

However, it is clear, that when the curve $\E$ is given, the sequence $(W_n)$ is predictable from two consecutive points $W_n$, $W_{n+1}$. In \cite{GutierrezIbeas}, Gutierrez and Ibeas showed that when the prime $p$ and $G$ are known, then the sequence $(W_n)$ is predictable even if just an approximation of $W_n$, $W_{n+1}$ are revealed (even if the curve $\E$ is private). These show that the security of the point sequence $(W_n)$ is not well-established. However these attacks use the assumption
that the prime and the curve (resp. the point $G$) are given which assumption is quite optimistic (in the viewpoint of the attacker). 

In our cryptographic settings, all the parameters, as the initial point $W_0=(x_0,y_0)$, the generator $G=(x_G,y_G)$, the parameters of the curve $A,B$ and the prime $p$ are assumed to be secret and just the output of the generator $x_1,x_2,\dots$ represented as non-negative integers, are used. The main contribution of this paper is that if some consecutive elements  are revealed, then one can compute in polynomial time (polynomial in $\log p$) an elliptic curve congruential generator (where the curve possibly defined over the rational or over a residue ring) such that the generated sequence is identical to $(x_n)$ in the revealed segment. It turns out that in practice, all the secret parameters, and thus the whole sequence $(x_n)$, can be computed from eight consecutive elements.

The result suggests that for cryptographic applications the elliptic curve congruential generator should be used with great care.

In Section~\ref{sec:background} we summarize some basic facts about elliptic curves. In Section~\ref{sec:main} we present the algorithm, and in Section \ref{sec:num} we discuss results of numerical tests.

\section{Background}\label{sec:background}

In this section we summarize some basic facts about elliptic curves. First we recall the definition of the group operation $\oplus$ of $\E$ defined over arbitrary field $k$. Then we extend the notion of elliptic curve to the case where it is defined over a ring.

\subsection{The group law on the curve}
Let $k$ be a field of characteristic different from $2,3$. Let $\E$ be an elliptic curve defined over $k$ given by an affine Weierstrass equation \eqref{eq:weierstrass} with $A,B\in k$, $4A^3+27B^2\neq 0$. The group operation $\oplus$ in $\E$ is defined in the following way.

\begin{definition}\label{def:1}
The operation $\oplus$ is defined over $\E$ as follows. If $P=(x_P,y_P)$ and $Q=(x_Q,y_Q)$ are in $\E(k)$, then
\[
 P\oplus Q=R=(x_R,y_R),
\]
where
\begin{enumerate}[(i)]
 \item \label{def:-i} if $x_P\neq x_Q$, then
 \[
  x_R=s^2-x_P-x_Q, \quad y_R=s(x_P-x_R)-y_P, \quad \text{where } s=\frac{y_Q-y_P}{x_Q-x_P};
 \]
 \item \label{def:-ii} if $x_P=  x_Q$ but $y_P\neq y_Q$, then $P\oplus Q=\mathcal{O}$;
 \item  if $P=Q$ and $y_P\neq 0$, then
 \[
  x_R=s^2-2x_P, \quad y_R=s(x_P-x_R)-y_P, \quad \text{where } s=\frac{3x_P^2+A}{2y_P};
 \]
 \item  if $P=Q$ and $y_P= 0$, then $P\oplus Q=\mathcal{O}$.
\end{enumerate}

\end{definition}
\subsection{Elliptic curves over $\Z_m$}

If $m$ is a composite integer with $\gcd(m,6)=1$, elliptic curve $\E$ can be also defined over $\Z_m$ via the \emph{projective Weierstrass equation}
\begin{equation*}
 y^2z=x^3+Axz^2+Bz^3
\end{equation*}
with $A,B\in\Z_m$, $\gcd(4A^3+27B^2,m)=1$. The $\Z_m$-rational points $\E(\Z_m)$ of $\E$ with projective coordinates can be represented as a triple $(x:y:z)$ such that $\gcd(m,x,y,z)=1$ (but not necessarily $z=0$ or $1$).

As in the field case, group operation can be defined on $\E$ (see \cite{lenstra,washington}) whose formulas correspond to Definition \ref{def:1}  if the divisor in (\ref{def:-i}) or (\ref{def:-ii})  is co-prime to $m$.

We remark that for integers $m_1,m_2$, $\gcd(m_1,m_2)=1$ we have
\[
 \E(\Z_{m_1m_2})\cong \E(\Z_{m_1}) \otimes \E(\Z_{m_2})
\]
as groups. Moreover, if $\E$ is an elliptic curve over $\Q$ defined by \eqref{eq:weierstrass} with integers $A,B$, then for  $m$ with $\gcd(4A^3+27B^2,m)=1$, the map $ \E(\Q)  \rightarrow \E(\Z_m)$ defined by
\[
 (x:y:z) \mapsto  (x \bmod m:y\bmod m:z\bmod m)  
\]
is a group homomorphism (where the representation $(x:y:z)$ is chosen as $x,y,z\in\Z$ and $\gcd(x,y,z)=1$).

Finally, for arbitrary integers $m_1,m_2$ (not necessarily co-primes), the the map $\E(\Z_{m_1m_2}) \rightarrow \E(\Z_{m_1})$ defined by
\[
 (x:y:z) \mapsto  (x \bmod m_2:y\bmod m_2:z\bmod m_2 ) 
\]
is a group homomorphism.

\section{Predicting the congruential generator on elliptic curve over rings}\label{sec:main}

Suppose we are given an initial segment $x_1,\dots, x_s$ of a sequence $(x_n)$ generated by an elliptic curve generator as non-negative integers. We would like to predict the remainder part of this sequence and specially, to compute the parameters of the generator, namely, the prime $p$, the parameters of the curve $A$, $B$ and the points $G$, $W_0$.

If for two different generators with primes $p$ and $q$, the revealed initial segments coincide, then the same initial segment is generated by an elliptic curve generator over $\Z_{p\cdot q}$. 
Clearly, in this case only the generator over the ring $\Z_{p\cdot q}$ is computable (without assuming the easiness of the integer factorization problem) and to recover the private parameters further  revealed elements are needed.

On the other hand, if 
the curve $\E$ is defined over $\Q$ with non-negative integers $A,B$ and $G,W_0\in\E(\Q)$ are points such that $x(iG\oplus W_0)$ are all integers for $i=1,\dots, s$, then there are infinitely many possible primes $p$ (and generators) exist, namely all large enough primes are suitable.

Thus our aim is to determine the most general elliptic curve generator (possibly over $\Q$ or over a ring $\Z_m$) which generates the same initial segment.

The following theorem shows that if at least seven initial values are revealed, then it can be computed 
a curve $\E$ over $\Q$ or over $\Z_m$ with $p\mid m$ and points $G$, $W_0$ such that these revealed values are the initial segment of a sequence generated by an elliptic curve generator with $\E$, $G$ and $W_0$. If more values are revealed, then a better approximation can be given to the generator (i.e. to the prime $p$).

\begin{thm}\label{thm:1}
There is an algorithm such that for given pairwise distinct non-negative integers $x_1,\dots, x_7$ generated by the elliptic curve congruential generator \eqref{eq:def-x} with prime number $p>3$, non-singular elliptic curve $\E=\E_{A,B}$ ($A,B\in\fp$), and points $W_0,G\in\ecp$ with $W_n\neq \pm G$ for $n=1,\dots, 6$, the algorithm computes a curve $\E_{\widetilde{A},\widetilde{B}}$ over $\Q$ or over $\Z_m$ with $p\mid m$ and a pair $(\widetilde{x_G}, \widetilde{(y_G^2)})$ in polynomial time (polynomial in $\log p$) such that if 
\begin{equation}\label{eq:gen}
\widetilde{x}_{n}=2\frac{\widetilde{x}_{n-1}^3+\widetilde{A} \widetilde{x}_{n-1}+\widetilde{B}+\widetilde{\left(y_G^2\right)}}{(\widetilde{x}_{n-1}-\widetilde{x_G})^2}-2(\widetilde{x}_{n-1}+\widetilde{x_G})-\widetilde{x}_{n-2}, \quad n=2,3,\dots,
\end{equation}
then $\widetilde{x}_{n}\equiv x_n \mod p$ whenever $W_n\neq \mathcal{O}$.
\end{thm}

We summarize the steps of computing in Algorithm \ref{alg:1}.

%

\renewcommand{\algorithmicrequire}{\textbf{Input:}}
\renewcommand{\algorithmicensure}{\textbf{Output:}}
\begin{algorithm}[ht]
\caption{Predicting the EC-LCG}\label{alg:1}

\begin{algorithmic}[1]
\Require non-negative integers $x_1,\dots, x_7$ generated by an elliptic curve congruential generator \eqref{eq:def-x}
\Ensure $\left(\widetilde{x},\widetilde{\left(y^2\right)},\widetilde{A},\widetilde{B},m\right)$ such that $p\mid m$, and if   $\widetilde{x}_n$ ($n=2,3,\dots$) generated by \eqref{eq:gen}, then $\widetilde{x}_n\equiv x_n \mod p$ whenever $W_n\neq \mathcal{O}$.
  \State define the vectors $\c_1,\c_2,\c_3,\c_4$ as the columns of the matrix $C$ in \eqref{eq:C} and $\mathbf{u}$ as \eqref{eq:u}
 \State $m\leftarrow\det(\c_1,\c_2,\c_3,\c_4,\mathbf{u})$.

 \If {$m=0$}
  \State write $\lambda_1\c_1+\lambda_2\c_2+\lambda_3\c_3+\lambda_4\c_4= \mu \mathbf{u}$ with $\lambda_1,\dots, \lambda_4,\mu\in\mathbb{Z}$, $\mu>0$
  \State $m\leftarrow(\lambda_1^2-\lambda_2\mu)/\gcd(\lambda_1^2,\mu)$
  
  \If {$m=0$}
    \State $ \widetilde{x}\leftarrow\frac{\lambda_1}{\mu}$, $\widetilde{\left(y^2\right)}\leftarrow \left(\frac{\lambda_1}{\mu}\right)^3+ \frac{\lambda_4}{2\mu}+\frac{\lambda_1\lambda_3}{2\mu^2}$,$\widetilde{A}\leftarrow\frac{\lambda_3}{\mu}$, $\widetilde{B}\leftarrow\frac{\lambda_4}{2\mu}-\frac{\lambda_1\lambda_3}{2\mu^2}$ 
  \EndIf
  
 \EndIf
\If {$m\neq 0$}
  \While {$\c_1$, $\c_2$, $\c_3$, $\c_4$ are linearly dependent modulo $m$}
  \State write $\lambda_1\c_1+\lambda_2\c_2+\lambda_3\c_3+\lambda_4\c_4\equiv 0 \mod m$ with $(\lambda_1,\lambda_2,\lambda_3,\lambda_4)\neq (0,0,0,0)$, 
  \State $m \leftarrow\gcd(m,\lambda_1,\lambda_2,\lambda_4,\lambda_5)$
 \EndWhile
 \State write $\lambda_1\c_1+\lambda_2\c_2+\lambda_3\c_3+\lambda_4\c_4\equiv \mu \mathbf{u}  \mod m$ with  $\mu>0$
  \If {$\gcd(m,\mu)>0$}
  \State $m\leftarrow m/\gcd(m,\mu)$
 \EndIf 
  
    \State write $\lambda_1\c_1+\lambda_2\c_2+\lambda_3\c_3+\lambda_4\c_4\equiv \mathbf{u}  \mod m$
    \If {$\lambda_1^2\not\equiv \lambda_2 \mod m$}
      \State $m\leftarrow\gcd(m,\lambda_1^2- \lambda_2 )$
    \EndIf
    \State $\widetilde{x}\leftarrow\lambda_1 \bmod m$, $\widetilde{\left(y^2\right)}\leftarrow \lambda_1^3+ \frac{\lambda_4+\lambda_1\lambda_3}{2} \bmod m$, $\widetilde{A}\leftarrow\lambda_3\bmod m$, 
    \State $\widetilde{B}\leftarrow\frac{\lambda_4-\lambda_1\lambda_3}{2}\bmod m$
 \EndIf
 \State \Return  $\left(\widetilde{x},\widetilde{\left(y^2\right)},\widetilde{A},\widetilde{B},m\right)$
\end{algorithmic}
\end{algorithm}

\begin{remark}
 The assumption that we are given the initial segment of the sequence $(x_n)$ is just a technical simplification. If any segment $x_{k},\dots, x_{k+6}$ of length seven are given, substituting $W_0$ by $W_{k-1}$ we get a generator whose initial values are $x_{k},\dots, x_{k+6}$.
\end{remark}

\begin{remark}
If we increase the number of revealed elements, then  the algorithm can be extended which provides  a better approximation $m$ to $p$ and hence a better approximation to $(x_{n})$, see Section \ref{sec:num}.
In practice, $8$ elements contain enough information to learn the exact value of $p$.  
\end{remark}


\begin{proof}[Proof of Theorem \ref{thm:1}]
Let us assume, that the integers $x_1,\dots, x_7$ generated by an elliptic curve congruential generator \eqref{eq:def-x} are given.

By \eqref{eq:def-x}, we can write
\[
 x_{i-1}=x(W_i\oplus (-G)), \quad  x_{i}=x(W_i), \quad x_{i+1}=x(W_i\oplus G), \quad i=2,\dots, 6,
\]
where $-G$ is the (additive) inverse of $G$: $-G=(x_G,-y_G)$.
By the addition law, by the assumption that $W_i\neq \pm G$ ($i=1,\dots, 6$) and by \eqref{eq:weierstrass} we have
\begin{align}\label{eq:x-sum}
 x_{i-1}+x_{i+1}&=\left(\frac{y_i+y}{x_i-x_G}\right)^2-x_i-x_G+\left(\frac{y_i-y}{x_i-x}\right)^2-x_i-x_G \\ \notag
                &=2\frac{y_i^2+y^2}{(x_i-x_G)^2}-2(x_i+x_G)\\
                &=2\frac{x_i^3+Ax_i+B+y_G^2}{(x_i-x_G)^2}-2(x_i+x_G),
                 \quad i=2,\dots, 6 \notag 
\end{align}
in $\fp$. Thus
\begin{multline*}
(x_i-x_G)^2(x_{i-1}+x_{i+1})\equiv 2(x_i^3+Ax_i+B+y_G^2)-2(x_i+x_G)(x_i-x_G)^2 \mod p, \\ i=2,\dots, 6 
\end{multline*}
i.e.,
\begin{multline}\label{eq:eqsys}
(2x_i^2+2x_i(x_{i-1}+x_{i+1}))x_G+(2x_i-(x_{i-1}+x_{i+1}))x_G^2+2x_iA+2B +2y_G^2  -2x_G^3\\
\equiv (x_{i-1}+x_{i+1})x_i^2 \mod p, \quad i=2,\dots, 6. 
\end{multline}
Put
\begin{equation}\label{eq:C}
C= 
\left(\!
\renewcommand{\arraystretch}{1.6}
\begin{array}{c@{\hspace{10pt}}c@{\hspace{10pt}}c@{\hspace{10pt}}c@{\hspace{10pt}}c@{\hspace{10pt}}c}
2x_2^2+2x_2 (x_1+x_3) & 2x_2-(x_1+x_3) & 2x_2 & 2 & 2 & -2\\
2x_3^2+2x_3 (x_2+x_4) & 2x_3-(x_2+x_4) & 2x_3 & 2 & 2 & -2\\
2x_4^2+2x_4 (x_3+x_5) & 2x_4-(x_3+x_5) & 2x_4 & 2 & 2 & -2\\
2x_5^2+2x_5 (x_4+x_6) & 2x_5-(x_4+x_6) & 2x_5 & 2 & 2 & -2\\
2x_6^2+2x_6 (x_5+x_7) & 2x_6-(x_5+x_7) & 2x_6 & 2 & 2 & -2\\
\end{array}\!
\right)\in\Q^{5 \times 6}
\end{equation}
and
\begin{equation}\label{eq:u}
\mathbf{u}= \left(\!
\renewcommand{\arraystretch}{1.6}
\begin{array}{c}
(x_1+x_3)x_2^2\\
(x_2+x_4)x_3^2\\
(x_3+x_5)x_4^2\\
(x_4+x_6)x_5^2\\
(x_5+x_7)x_6^2\\
\end{array}\!
\right)\in\Q^{5}.
\end{equation}

Write $C=(\c_1,\dots, \c_6)$ with $\c_1,\dots, \c_6 \in\Z^5$. Then we have
\begin{lemma}\label{lemma:1}
Having the same assumption as in Theorem \ref{thm:1}, the columns $\c_1$, $\c_2$, $\c_3$, $\c_4$ are linearly independent over $\fp$. 
\end{lemma}

Assuming Lemma \ref{lemma:1}, the matrix $C$ has rank 4 over $\fp$, and by \eqref{eq:eqsys} the congruence
\begin{equation}\label{eq:C-modp}
 C \cdot \mathbf{e}
 \equiv\mathbf{u}
\mod p
\end{equation}
has the  solution  $\mathbf{e}=(x_G,x_G^2, A,B,y_G^2,x_G^3)^T$.

The algorithm looks for an
integer $m$ such that the congruence \eqref{eq:C-modp} has a solution $\mathbf{e}=(e_1,\dots, e_6)^T$ modulo $m$ with the additional restriction $e_1^2 \equiv e_2 \mod m$. During the algorithm we will always have  $p\mid m$. If $m$ takes a finite value, then we can suppose that $\gcd(m,6)=1$, since $p>3$. Moreover, $a\equiv b \mod m$ for $m=0$, means that $a=b$ as rationals.

Since the congruence \eqref{eq:C-modp} has solutions, $\det(\c_1,\c_2,\c_3,\c_4,\mathbf{u})\equiv 0 \mod p$.
First, assume that $\det(\c_1,\c_2,\c_3,\c_4,\mathbf{u})\neq 0$ and put $m=\det(\c_1,\c_2,\c_3,\c_4,\mathbf{u})$.
If $\c_1$, $\c_2$, $\c_3$, $\c_4$ are linearly dependent modulo $m$, say, $\lambda_1\c_1+\lambda_2\c_2+\lambda_3\c_3+\lambda_4\c_4\equiv 0 \mod m$ with $(\lambda_1,\lambda_2,\lambda_3,\lambda_4)\neq (0,0,0,0)$, then we have $p\mid \gcd(\lambda_1,\allowbreak \lambda_2,\allowbreak \lambda_3,\lambda_4)$. Replacing $m$ to $\gcd(m,\lambda_1,\lambda_2,\lambda_4,\lambda_5)$ we reduced $m$ with maintaining the property $p\mid m$. Iterating this step, we can assume that $\c_1$, $\c_2$, $\c_3$, $\c_4$ are linearly independent modulo $m$. 

By the vanishing of the determinant, $\c_1$, $\c_2$, $\c_3$, $\c_4$, $\mathbf{u}$ are linearly dependent modulo $m$: $\lambda_1\c_1+\lambda_2\c_2+\lambda_3\c_3+\lambda_4\c_4\equiv \mu \mathbf{u}  \mod m$. By the independence of $\c_1,\c_2,\c_3,\c_4$, $\mu \not \equiv 0 \mod m$ and if $\mu$ is minimal and positive, the coefficients $\lambda_1$, $\lambda_2$, $\lambda_3$, $\lambda_4$, $\mu$ are unique. If $\gcd(m,\mu)>1$, replacing $m$ to $m/\gcd(m,\mu)$ we can assume $\mu=1$, i.e. $\lambda_1\c_1+\lambda_2\c_2+\lambda_3\c_3+\lambda_4\c_4\equiv \mathbf{u}  \mod m$ with unique $\lambda_1$, $\lambda_2$, $\lambda_3$, $\lambda_4$. Finally, all solutions of the congruence
\begin{equation}\label{eq:eqsys-modm}
  C \cdot \mathbf{e}
 \equiv\mathbf{u}
\mod m
\end{equation}
can be expressed as
\begin{equation}\label{eq:eqsys-gen}
 e_1=\lambda_1, \quad e_2=\lambda_2, \quad e_3=\lambda_3, \quad e_4+e_5-e_6=\lambda_4.
\end{equation}
If $\lambda_1^2\not\equiv \lambda_2 \mod m$, we have to replace $m$ to $\gcd(m,\lambda_1^2- \lambda_2 )$.

Next, consider the case when  $\det(\c_1,\c_2,\c_3,\c_4,\mathbf{u})= 0$ (over $\Q$). Now, the equation
\[
   C \cdot \mathbf{e} = \mathbf{u}
\]
has solutions over $\Q$. By Lemma \ref{lemma:1}, $\c_1,\c_2,\c_3,\c_4$ are linearly independent over $\fp$ and thus over $\Q$. As before 
\begin{equation}\label{eq:solutionQQ}
\lambda_1\c_1+\lambda_2\c_2+\lambda_3\c_3+\lambda_4\c_4= \mu \mathbf{u}, \quad \gcd(\lambda_1,\lambda_2,\lambda_3,\lambda_4,\mu)=1, \quad \mu>0   
\end{equation}
with unique integers $\lambda_1,\allowbreak \lambda_2,\lambda_3, \allowbreak \lambda_4,\mu$. Moreover, every  integer solution $(\lambda_1',\allowbreak \lambda_2',\allowbreak \lambda_3',\allowbreak \lambda_4',\mu')$ of \eqref{eq:solutionQQ} has the form $\gamma\cdot (\lambda_1,\lambda_2,\lambda_3,\lambda_4,\mu)$ with some $\gamma\in\Z$.   Since $(x,x^2,A,B+y^2-x^3,1)$ is a solution modulo $p$ we have for $\gcd(p,\mu)=1$, that $(\lambda_1/\mu)^2\equiv \lambda_2/\mu \mod p$. Put $m=(\lambda_1^2-\lambda_2\mu)/\gcd(\lambda_1^2,\mu)$ so $p\mid m$. If $m\neq 0$, write the solutions of \eqref{eq:eqsys-modm} in the same form as \eqref{eq:eqsys-gen}.

In both cases, if $m\neq 0$ write 
\[
\widetilde{x_G}=\lambda_1,\quad \widetilde{\left(y_G^2\right)}= \lambda_1^3+ \frac{\lambda_4+\lambda_1\lambda_3}{2},\quad 
\widetilde{A}=\lambda_3,\quad \widetilde{B}=\frac{\lambda_4-\lambda_1\lambda_3}{2} \quad 
\text{over } \Z_m.
\]
Clearly, if $\widetilde{G}=(\widetilde{x_G},\widetilde{y_G})$ with a $\widetilde{y_G}\in\Z_m[\zeta]/(\zeta^2-\widetilde{\left(y_G^2\right)})$,  $\widetilde{y_G}^2=\widetilde{\left(y_G^2\right)}$, then $\widetilde{G}\in\E_{\widetilde{A},\widetilde{B}}$ over $\Z_m[\zeta]/(\zeta^2-\widetilde{\left(y_G^2\right)})$. Moreover, the vector
\[
\e=(\widetilde{x_G}, \widetilde{x_G}^2,  \widetilde{A},\widetilde{B},\widetilde{\left(y_G^2\right)},\widetilde{x_G}^3 )^T \in\Z_m^6                                                                                                                                                                                                                                                                                                                                                                                                                                        
\]
is a solution of \eqref{eq:eqsys-modm}. If $\widetilde{y}_1\in\Z_m[\xi]/(\xi^2-x_1^3-\widetilde{A}x_1-\widetilde{B})$ is an element such that $\widetilde{W_1}=(x_1,\widetilde{y}_1)$ is on the curve $\E_{\widetilde{A},\widetilde{B}}$, then writing $\widetilde{W_0}=\widetilde{W_1}\oplus(-\widetilde{G})$, the sequence $(\widetilde{x}_n)$  generated by the elliptic curve congruence generator (with $\widetilde{W_0}$, $\widetilde{G}$) satisfies \eqref{eq:gen}
so $\widetilde{x}_1, \widetilde{x}_2,\dots$ are in $\Z_m$ and thus $\widetilde{x}_i=x_i$ ($i=1,2,\dots, 7$) as integers.

Finally, if $m=0$, write 
\[
\widetilde{x_G}=\frac{\lambda_1}{\mu},\quad \widetilde{\left(y_G^2\right)}= \left(\frac{\lambda_1}{\mu}\right)^3+ \frac{\lambda_4}{2\mu}+\frac{\lambda_1\lambda_3}{2\mu^2},\quad 
\widetilde{A}=\frac{\lambda_3}{\mu},\quad \widetilde{B}=\frac{\lambda_4}{2\mu}-\frac{\lambda_1\lambda_3}{2\mu^2} 
\]
over $\Q$.
Then $\widetilde{G}=(\widetilde{x_G},\widetilde{y_G})\in\E_{\widetilde{A},\widetilde{B}}$ over $\Q\left(\sqrt{ \widetilde{\left(y^2\right)} }\right)$ and the integer sequence $x_1\dots, x_7$ is generated by $\widetilde{G}$. The set of possible primes $p$ are those ones which $p\nmid \mu$ and $p>\max\{ x_i: \ i=1,\dots, 7\}$.
\end{proof}

Finally, it is remain to prove Lemma \ref{lemma:1}.

\begin{proof}[Proof of Lemma \ref{lemma:1}]
Clearly, it is enough to show that the vectors
\begin{equation*}
\renewcommand{\arraystretch}{1.6}
\mathbf{v}_1=
\left(\!
\begin{array}{c}
 x_2^2+x_2(x_1+x_3)  \\ x_3^2+x_3(x_2+x_4)\\ x_4^2+x_4(x_3+x_5)\\x_5^2+x_5(x_4+x_6) \\ x_{6}^2+x_{6}(x_{5}+x_{7})
\end{array}\!
\right)\!,
\mathbf{v}_2=
\left(\!
\begin{array}{c}
 x_1+x_3  \\ x_2+x_4\\ x_3+x_5\\x_4+x_6 \\ x_{5}+x_{7}
\end{array}\!
\right)\!,
\mathbf{v}_3=
\left(\!
\begin{array}{c}
 x_2\\ x_3\\ x_4\\x_5 \\x_{6}
\end{array}\!
\right)\!,
\mathbf{v}_4=
\left(\!
\begin{array}{c}
 1  \\ 1\\ 1\\1 \\ 1
\end{array}\!
\right)
\end{equation*}
are linearly independent over $\fp$. 

Suppose to the contrary that there are $\alpha_1,\dots, \alpha_4\in\fp$, $(\alpha_1,\dots, \alpha_4)\neq (0,\dots, 0)$ such that
\begin{equation*}
 \alpha_1\mathbf{v}_1+\alpha_2\mathbf{v}_2+\alpha_3\mathbf{v}_3+\alpha_4\mathbf{v}_4=0
\end{equation*}
i.e.,
\begin{equation}\label{eq:linComb-i}
\alpha_1( x_i^2+x_i(x_{i-1}+x_{i+1}))+\alpha_2(x_{i-1}+x_{i+1})+\alpha_3 x_i+\alpha_4=0, \quad i=2,\dots, 6. 
\end{equation}
Substituting \eqref{eq:x-sum} to \eqref{eq:linComb-i} we have
\begin{multline*}
 \alpha_1\left( x_i^2+2x_i\left( \frac{x_i^3+Ax_i+B+y_G^2}{(x_i-x_G)^2}-(x_i+x_G)  \right)    \right)\\+2\alpha_2\left( \frac{x_i^3+Ax_i+B+y_G^2}{(x_i-x_G)^2}-(x_i+x_G)  \right) +\alpha_3 x_i+\alpha_4=0, \quad i=2,\dots, 6.
\end{multline*}
Clearing the denominator we get
\begin{multline*}
\alpha_1\left( x_i^2(x_i-x_G)^2+2x_i\left(x_i^3+Ax_i+B+y_G^2-(x_i+x_G)(x_i-x_G)^2  \right)    \right)\\
+2\alpha_2\left( x_i^3+Ax_i+B+y_G^2-(x_i+x_G)(x_i-x_G)^2  \right)\\
+\alpha_3 x_i(x_i-x_G)^2+\alpha_4(x_i-x_G)^2=0
\end{multline*}
for $i=2,\dots, 6$, which means that the polynomial $F(X)\in\fp[X]$
\begin{align*}
 F(X)&=\\
 &\alpha_1\left( X^2(X-x_G)^2+2X\left(X^3+AX+B+y_G^2-(X+x_G)(X-x_G)^2  \right)\right)\\
 +&2\alpha_2\left( X^3+AX+B+y_G^2(X+x_G)(X-x_G)^2  \right) +\alpha_3 X(X-x_G)^2\\
 +& \alpha_4(X-x_G)^2
\end{align*}
has at least five zeros: $x_i$, $i=2,\dots 6$.

Write $F(X)$ into the following form
\begin{multline}\label{eq:F-reduced}
 F(X)=\alpha_1 \left(X^4 + f_4(X)\right)+\alpha_3 \left(X^3 + f_3(X)\right)+(2\alpha_2x+\alpha_4)X^2\\+ (\alpha_2(2x^2+2A)-2\alpha_4x)X+\alpha_2(-2x^3+2B+2y^2)+\alpha_4x^2,
\end{multline}
where $f_3,f_4\in\fp[X]$ with $\deg f_3<3$  and $\deg f_4<4$.

Since $\deg F\leq 4$, we must have that $F(X)$ is the zero polynomial. In this case we have
\begin{equation}\label{eq:zero-coefficients}
 \alpha_1=0, \quad \alpha_3=0, \quad 2\alpha_2x+\alpha_4=0.
\end{equation}
Then the coefficients of $X$ and 1 in \eqref{eq:F-reduced} are
\[
\alpha_2(3x_G^2+A)=0, \quad \alpha_2(2x_G^3-B-y_G^2)=0.
\]

By \eqref{eq:zero-coefficients} and $(\alpha_1,\dots, \alpha_4)\neq (0,\dots, 0)$ we also have 
$\alpha_2,\alpha_4\neq 0$, thus
\begin{equation*}
 3x_G^2+A=0, \quad 2x_G^3-B-y_G^2=0,
\end{equation*}
whence using \eqref{eq:weierstrass} we get
\begin{equation*}
 2y_G^2=x(3x_G^2+A)-(x_G^3+Ax_G+B-y_G^2)-(2x_G^3-B-y_G^2)=0,
\end{equation*}
thus
\begin{equation*}
 y_G=0,\quad 3x_G^2+A=0.
\end{equation*}
Since $G=(x_G,y_G)\in\ecp$, we get that $x_G$ is a multiple root of the right hand side of \eqref{eq:weierstrass}, which contradict that the discriminant of the curve is non-zero. 
\end{proof}

\section{Numerical tests}\label{sec:num}
I have implemented the algorithm of Theorem \ref{thm:1} in SAGE. The algorithm have been tested for 1000 random examples of generators with 500-bit primes $p$. For seven revealed sequence elements, the algorithm computed the exact values of the parameters ($p$, $A$, $B$, $G$, $W_0$) in 95,2\% of the cases. In the remainder cases, the algorithm provided a composite integer $m$ and parameters $\widetilde{A}$, $\widetilde{B}$, $\widetilde{G}$, $\widetilde{W_0}$ such that $p\mid m$ and $\widetilde{A}\equiv A$, $\widetilde{B}\equiv B$, $\widetilde{G}\equiv G$, $\widetilde{W_0}\equiv W_0 \mod p$. 

If the number of revealed sequence elements increases, then the algorithm can be modified to become more effective. Namely, if there are eight revealed sequence elements, then applying the algorithm for the first and the last seven elements, two approximation $(m_1,A_1,B_1,G_1)$ and $(m_2,A_2,B_2,G_2)$ are provided. Putting $m=\gcd(m_1,m_2, A_1-A_2,B_1-B_2, x(G_1)-x(G_2))$, $m$ is a better approximation of $p$. Modifying the SAGE program in this way, the algorithm was successful in 100\% of the cases.

\section{Final remarks, open questions}
In this paper we showed that the sequence $x_n=x(nG\oplus W_0)$ ($n=1,2,\dots$) is highly predictable if at least seven consecutive elements are revealed. A natural question is whether one can reduce the number of revealed elements. 

In the literature, the distribution and the linear complexity profile of the general sequence $f(nG\oplus W_0)$ with $f\in\ffp$ have been also studied (where $\ffp$ is the function field of $\ecp$), see \cite{BD,Clegendre,CGP,CLX,HessShparlinski,HHF,Liu,LWZ,merai-1,merai-2}. The predictability of these sequences could be handle for individual functions $f$, but it is not clear whether there is a universal algorithm for all function $f$ (or at least all function $f$ with small degree). 

An other possible question connected to the result is how much information about $x(nG\oplus W_0)$  we really need to recover all the private parameters. 
Gutierrez and Ibeas \cite{GutierrezIbeas} showed that when the prime $p$ and $G$ are known, then the sequence $(W_n)$ is predictable even if just an approximation of $W_n$, $W_{n+1}$ are revealed. However, the problem is still open in the case when the prime is unknown. 
One can also  consider this problem over arbitrary (not prime) finite fields. 
Namely, if the curve is defined over a finite field $\mathbb{F}_{p^s}$ with degree $s>1$, then one can define an integer sequence as $(x_1(nG\oplus W_0),\dots, x_r(nG\oplus W_0))$ $n=1,2,\dots$   where $r\leq s$ and $x_1(P),\dots, x_s(P)$ are the coordinates of $x(P)$ with respect to a fixed basis of $\mathbb{F}_{p^s}$ over $\fp$. Is it possible to recover the whole sequence from an initial segment, at least when the degree $s$ is fixed?  Clearly, the most interesting case when $p=2$ and the generator builds a binary sequence.

\section*{Acknowledgments}
The author would like to thank Harald Niederreiter for suggesting the problem. The authors would also like to thank the reviewers for their perceptive and useful comments that significantly improved the work.

The author is partially supported by the Austrian Science Fund FWF Project F5511-N26 which is part of the Special Research Program "Quasi-Monte Carlo Methods: Theory and Applications" and by Hungarian National Foundation for Scientific Research,Grant No. K100291.

\end{document}